\documentclass{article}
\usepackage{amsmath}
\usepackage{amsfonts}
\usepackage{amsthm}
\usepackage{comment}
\usepackage[utf8]{inputenc}
\usepackage[toc]{appendix}
\usepackage{algorithm}
\usepackage{algpseudocode}
\usepackage{wasysym,utfsym}
\usepackage{arydshln}

\newtheorem{lemma}{Lemma}
\newtheorem{theorem}{Theorem}

\newtheorem{corollary}{Corollary}
\newtheorem{definition}{Definition}

\newtheorem{example}{Example} 

\newcommand{\etal}{\textit{et al. }}

\usepackage[
backend=biber,
style=alphabetic,
sorting=nyt
]{biblatex}

\title{On the Uniqueness of Balanced Complex Orthogonal Design}
\author{
Yiwen Gao\footnote{School of Computer Science, Fudan University, Shanghai 200433, China. Email: ywgao21@m.fudan.edu.cn},
Yuan Li\footnote{School of Computer Science, Fudan University, Shanghai 200433, China. Email: yuan\_li@fudan.edu.cn} and
Haibin Kan\footnote{Shanghai Key Laboratory of Intelligent Information Processing, School of Computer Science, Fudan University, Shanghai 200433, China;
Shanghai Engineering Research Center of Blockchain, Shanghai 200433, China;
Yiwu Research Institute of Fudan University, Yiwu City 322000, China. Email: hbkan@fudan.edu.cn}}
\date{}

\begin{document}

\maketitle

\abstract{
\emph{Complex orthogonal designs} (CODs) play a crucial role in the construction of space-time block codes. Their real analog, real orthogonal designs (or equivalently, sum of squares composition formula) have a long history. Adams \etal (2011) introduced the concept of \emph{balanced complex orthogonal designs} (BCODs) to address practical considerations. BCODs have a constant code rate of $1/2$ and a minimum decoding delay of $2^m$, where $2m$ is the number of columns. Understanding the structure of BCODs helps design space-time block codes, and it is also fascinating in its own right.

We prove, when the number of columns is fixed, all (indecomposable) balanced complex orthogonal designs (BCODs) have the same parameters $[2^m, 2m, 2^{m-1}]$, and moreover, they are all equivalent.
}

\section{Introduction}

The investigation of orthogonal designs, also known as the sum of squares composition formula, has a rich historical background. The main problem is, for which $[p, n, k]$, does there exist a $p \times n$ matrix $G$, where each nonzero entry is a linear combination of indeterminates $x_1, x_2, \ldots, x_k$ with complex coefficients such that $G^T G = (|x_1|^2 + \ldots + |x_k|^2) I_n$. A $[p, n, k]$ orthogonal design is equivalent to a $[n, k, p]$ composition formula
\[
(x_1^2 + \ldots + x_n^2) (y_1^2 + \ldots + y_k^2) = (z_1^2 + \ldots + z_p^2),
\]
where $X=(x_1, \ldots, x_n)$ and $Y = (y_1, \ldots, y_k)$ are indeterminates, and each $z_k$ is bilinear in $X$ and $Y$. 
Orthogonal designs have an intimate connection with many fields in mathematics, including combinatorics, topology, representation theory, etc. For an in-depth study of this subject matter, Shapiro's and Seberry's books offer exceptional resources \cite{Shapiro2011, Seberry2017}.

Results regarding \emph{complex orthogonal designs (CODs)} were rare in comparison to (real) orthogonal designs. A. V. Geramita and J. M. Geramita were the first to propose the definition of complex orthogonal designs in \cite{Geramita1978}, focusing exclusively on square CODs. Alarcon and Yiu proposed the definition of non-square CODs in the context of \emph{Hermitian sum of squares formula} \cite{AY1993}. A Hermitian $[n, k, p]$ formula is
$$
(|x_1|^2 + \ldots + |x_n|^2) (|y_1|^2 + \ldots + |y_k|^2) = (|z_1|^2 + \ldots + |z_p|^2),
$$
where $X=(x_1, \ldots, x_n)$, $Y = (y_1, \ldots, y_k)$ are \emph{complex} indeterminates, and each $z_k$ is bilinear in $(X, \overline{X})$ and $Y$. In fact, a Hermitian sum of squares formula $[n, k, p]$ formula is equivalent to a $[p, n, k]$ COD. For further insights into this connection, we recommend consulting Shapiro's book \cite{Shapiro2011}, which explores this topic extensively while also providing other intriguing results concerning sums of squares. 

Due to its application in \emph{space-time block codes}, the study of CODs was revived. In 1999, Tarokh \etal introduced space-time block codes, a method for transmission using multiple antennas \cite{Tarokh1999}. For complex signal constellations, the definition of complex orthogonal design arises naturally. Roughly speaking, a $[p, n, k]$ COD $G$ is a $p \times n$ matrix, where each nonzero entry is a complex linear combination of complex indeterminates $z_1, z_2, \ldots, z_k$ and their conjugates, such that $G^H G = (|z_1|^2 + \ldots + |z_k|^2) I_n$, where $G^H$ denotes the Hermitian transpose of $G$. In space-time block codes, $n$ corresponds to the number of antennas, $p$ the decoding delay, and $k/p$ the transmission rate. For readers seeking additional information on space-time codes and wireless communications, we recommend consulting the book \cite{Wysocki2007}.

At that time, two burning problems were to determine the \emph{maximal rate} $k/p$ for any number of antennas $n$, and the \emph{minimum decoding delay} $p$ when the rate is maximal.  Liang completely solved the maximal rate problem \cite{Liang2003}. Liang proved $\frac{k}{p} \le \frac{m+1}{2m}$ when $n = 2m$ or $2m-1$. Furthermore, Liang provided an algorithm that generates CODs achieving the rate. For CODs with linear combination, Haiquan Wang and Xiang-gen Xia proved that when $n\ge3$, the rate is upper bounded by $\frac{3}{4}$ \cite{Wang2003}. This bound is reachable when $n=3,4$ 
\cite{Tarokh1999}. In two papers \cite{Adams2007, Adams2009}, Adams \etal solved the minimum decoding delay problem by proving
\[
p \ge \begin{cases}
 {2m \choose m-1} & \text{if $n \not\equiv 2 \pmod 4$}\\
 2{2m \choose m-1} & \text{if $n \equiv 2 \pmod 4$},
 \end{cases}
\]
where $n = 2m$ or $2m-1$. The tightness of the above lower bound on $p$ was already shown by constructions in \cite{Liang2003, lu2005closed}. Li \etal determined all possible structures for CODs without $2 \times 2$ submatrix 
$\begin{pmatrix}
	\pm z_j & 0 \\
	0  & \pm z^*_j \\
\end{pmatrix},
$
which contains all the CODs with maximal rate, and some others. Note that the results mentioned hereafter are proved in the combinatorial setting, where each nonzero entry is $\pm z_i$ or $\pm z_i^*$.

In addition to rate ($C_1$) and decoding delay ($C_2$), when designing CODs for practical applications, researchers consider several other criterions including
\begin{itemize}
	\item ($C_3$) transceiver signal linearization: linearization is achieved if the COD is conjugation-separated, i.e., nonzero entries in a row are all conjugated or non-conjugated.
	\item ($C_4$) peak-to-average power ratio (PAPR): zero-free CODs are desirable to achieve low PAPR, which reduces the need to switch on and off antennas.
	\item ($C_5$) power balance: each variable appears the same number of times in each column.
	\item ($C_6$) no irrational coefficients: it eliminates floating point multiplication.
	\item ($C_7$) no linear processing: every nonzero entry is $\pm z_k$ or $\pm z_k^*$. It reduces signaling and decoding complexity.
\end{itemize}
Regarding these characteristics, we send interested readers to \cite{Adams2011} for a detailed discussion. Driven by these practical considerations, Adams \etal proposed the definition of \emph{balanced complex orthogonal design} (BCOD) \cite{Adams2011}. BCODs have a constant rate of $\frac{k}{p} = \frac{1}{2}$ and a minimum  decoding delay of $p = 2^m$, where $n = 2m$. In addition, they provided a method to generate a rate 1/2 COD with no zero entries ($C_4$) through a BCOD and its zero-masking row companion matrix. The properties of these two types of CODs are outlined in Table~\ref{table: properties of BCOD}: BCOD performs well with respect to $C_3$ and $C_5$-$C_7$; the modified rate 1/2 COD performs well with respect to $C_4$-$C_7$.

\begin{table}[!htbp]\label{table: properties of BCOD}
\centering
\begin{tabular}{|l|c|c|}
    \hline
    && the modified rate 1/2 \\&BCOD&COD in \cite{Adams2011}\\
    \hline
    $C_1$: Rate & 1/2 & 1/2\\
    \hline
    $C_2$: Decoding delay & $2^{m}$($2m$ columns) & $2^{m}$($2m$ columns)\\
    \hline
    $C_3$: Transceiver signal&&\\ linearization & \checkmark & \scalebox{0.75}{\usym{2613}}\\
    \hline
    $C_4$: Peak-to-average power&&\\ ratio & half zero entries & zero-free\\
    \hline
    $C_5$: Power balance & \checked & \checked\\
    \hline
    $C_6$: No irrational coefficients & \checked & \checked\\
    \hline
    $C_7$: No linear processing & \checked & \checked\\
    \hline 
\end{tabular}
\caption{Properties of BCOD and the rate 1/2 COD in \cite{Adams2011}}
\end{table}

When $n \not\equiv 0 \pmod 8$, Adams \etal proved a lower bound $2^m$ on $p$ by reducing BCODs to real orthogonal designs and applying classical results on real orthogonal designs \cite{Adams2011}. The final case that $n \equiv 0 \pmod 8$ was settled by Liu \etal \cite{Liu2015} by a purely combinatorial argument. Independently, Das \cite{Das2016} gave proof of $p \ge 2^m$ and provided a construction that achieves the bound. In addition, Das proved that BCODs with the same parameters are equivalent \emph{ignoring} signs and conjugations. In addition to their practical nature, BCODs are elegant combinatorial objects. So understanding the structure of BCODs is a meaningful task.

\subsection{Our results}

In this paper, we prove BCODs are unique (up to equivalence) when the number of columns is fixed. Specifically, when $n = 2m$, all indecomposable BCODs have parameters $[2^m, 2m, 2^{m-1}]$, and they are all equivalent.

The overall proof strategy is to show that after removing two columns, a BCOD is the concatenation of two smaller BCODs. Using mathematical induction, we prove that all (indecomposable) BCODs are of the same parameter $[2^m, 2m, 2^{m-1}]$.

As for the uniqueness, we prove by induction on a stronger statement, that is, any BCOD can be transformed into a canonical form using equivalence operations \emph{without} any column operations (including column permutations and column negations). Note that there are two BCODs (of smaller size), one on top of the other. As such, column operations will apply to both of them simultaneously. If one can transform each of them into canonical forms without column operations, then one can transform \emph{both} of them into canonical forms at the same time.

In terms of techniques, our proof introduces the following new ingredients:
\begin{itemize}
    \item We prove that any (indecomposable) BCOD is the concatenation of two smaller BCODs, after removing two columns.
    \item We prove that any (indecomposable) BCOD can be transformed into a canonical form without column operations.
    \item We use the locality argument to prove that the remaining two columns (that were previously removed) are uniquely determined (due to orthogonality).
\end{itemize}
We believe our techniques have the potential to go beyond BCODs. It is likely that one can apply similar arguments to broader classes of CODs.

Despite their deep practical nature, we believe CODs are nice combinatorial objects worth studying and possess much nicer structures than the real ones.

%Results in this paper complement a classification result in \cite{Li2012}, where 

\section{Preliminaries}

\emph{Balanced complex orthogonal design} (BCOD) is a class of rate $1/2$ complex orthogonal design (COD) introduced by Adams \etal \cite{Adams2011}. In this section, we review the relevant definitions and some known facts about CODs and BCODs.

\begin{definition}
A \emph{complex orthogonal design (COD)} $G[p, n, k]$ is a $p \times n$ matrix, whose nonzero entry is either $\pm z_i$ or $\pm z_i^*$, $i = 1, 2, \ldots, k$, where $z_1, z_2, \ldots, z_k$ are complex indeterminates, such that
\[
G^H G = (|z_1|^2 + |z_2|^2 + \ldots + |z_k|^2) I_n,
\]
where $G^H$ denotes the Hermitian transpose of $G$.
\end{definition}

For notational convenience, let $[z_i]$ denote entry $z_i$, $-z_i$, $z^*_i$, or $-z^*_i$. By the definition of COD, it is clear that for every $i$, entry $[z_i]$ must appear in each column exactly once, and any two columns are orthogonal.

In some literature, linear processing is allowed in the definition of CODs, that is, each nonzero entry is a complex linear combination of $z_1, z_2, \ldots, z_k$ and their conjugations, which are called \emph{CODs with linear processing}. However, CODs with linear processing are not in the scope of our paper; we study CODs without linear processing using combinatorial methods.

\begin{definition} \label{def:equiv}(Equivalence operations) The following operations on a COD are called \emph{equivalent operations}:
\begin{itemize}
\item Rearrange the order of rows (``row permutation'').
\item Multiply any row by $-1$ (``row negation'').
\item Rearrange the order of columns (``column permutation'').
\item Multiply any column by $-1$ (``column negation'').
\item Change the index of all instances of a certain variable (``instance renaming'').
\item Negate all instances of a certain variable (``instance negation'').
\item Conjugate all instances of a certain variable (``instance conjugation'').
\end{itemize}

Two CODs are called \emph{equivalent} if one can be changed into the other through equivalence operations.
\end{definition}

One can easily verify that equivalence operations preserve orthogonality. In other words, after applying an arbitrary sequence of equivalence operations to a $[p, n, k]$ COD, we end up with another $[p, n, k]$ COD.

The main problem of complex orthogonal design is: for which $p, n, k$, does a $[p, n, k]$ COD exist? Furthermore, what are the different $[p, n, k]$ CODs up to equivalence?

Unfortunately, the main problem is wide open; we are only able to resolve it for a few restricted classes of CODs. In this paper, our goal is to answer it for BCODs.

\begin{definition}
\label{def:Bi_form}
 ($B_i$ form \cite{Adams2007}) Let $G$ be a $[p, n, k]$ COD. We say that $G$ is \emph{in $B_i$ form} if, after equivalence operations except for column permutation, $G$ contains the following submatrix
\[
	B_i=
	\left(
	\begin{array}{c | c}
	\begin{matrix}
	z_i & 0 & \dots & 0 \\
	0 & z_i & \dots & 0 \\
	\vdots & \vdots & \ddots & \vdots \\
	0 & 0 & \dots & z_i
	\end{matrix}
	& M_i \\
	\hline
	-M_i^H 
	&\begin{matrix}
	z_i^* & 0 & \dots & 0 \\
	0 & z_i^* & \dots & 0 \\
	\vdots & \vdots & \ddots & \vdots \\
	0 & 0 & \dots & z_i^*
	\end{matrix}
	\end{array}
	\right),
\]
where $M_i$ is a submatrix whose nonzero entry is either $\pm z_i$ or $\pm z_i^*$, $i = 1, 2, \ldots, k$. 
We call the submatrix of $G$ that is equivalent to $B_i$ (defined above) \emph{$B_i$ form submatrix} of $G$. 
\end{definition}

Observe that a COD can be transformed into $B_i$ form, for any $i$, if column permutations are allowed; in Definition \ref{def:Bi_form}, column permutations are not allowed. For example, COD 
$\begin{pmatrix}
	z_1 & z_2 \\
	-z_2^* & z_1^*
\end{pmatrix}$
is in $B_1$ form, but not in $B_2$ form.

The following example shows two equivalent CODs in $B_1$ form and $B_2$ form respectively.
%The following example shows the same(according to equivalence operations defined in definition~\ref{def:Bi_form}) $[8,5,4]$COD in $B_1$ and $B_2$ forms. 

\begin{example}\label{emample:Bi_form} $[8, 5, 4]$ COD.
%Two equivalent $[8,5,4]$ CODs in $B_1$ and $B_2$ forms.
\[
G_1 = 
\left(
\begin{array}{c c c : c c}
z_1 & 0 & 0 & 0 & z_2\\
0 & z_1 & 0 & -z_2 & 0\\
0 & 0 & z_1 & -z_3 & -z_4\\
\hdashline
0 & z_2^{*} & z_3^{*} & z_1^{*} & 0\\
-z_2^{*} & 0 & z_4^{*} & 0 & z_1^{*}\\
\hdashline
-z_3^{*} & -z_4^{*} & 0 & 0 & 0\\
z_4 & -z_3 & z_2 & 0 & 0\\
0 & 0 & 0 & z_4 & -z_3^{*}
\end{array}
\right)
\]
\[
G_2 =
\left(
\begin{array}{c c c : c c}
z_2 & 0 & 0 & 0 & z_1\\
0 & z_2 & 0 & -z_1 & 0\\
0 & 0 & z_2 & -z_3 & z_4\\
\hdashline
0 & z_1^{*} & z_3^{*} & z_2^{*} & 0\\
-z_1^{*} & 0 & -z_4^{*} & 0 & z_2^{*}\\
\hdashline
-z_3^{*} & z_4^{*} & 0 & 0 & 0\\
-z_4 & -z_3 & z_1 & 0 & 0\\
0 & 0 & 0 & -z_4^{*} & -z_3^{*}
\end{array}
\right)
\]
$G_1$ is in $B_1$ form, and $G_2$ is in $G_2$ form. We can change $G_1$ into $G_2$ by
\begin{itemize}
    \item permutating columns: $\begin{pmatrix}
        1 & 2 & 3 & 4 & 5\\
        5 & 4 & 3 & 2 & 1
    \end{pmatrix},$
    \item permutating rows: 
    $\begin{pmatrix}
        1 & 2 & 3 & 4 & 5 & 6 & 7 & 8\\
        1 & 2 & 7 & 4 & 5 & 8 & 3 & 6
    \end{pmatrix}$, and
    \item negating row $2$ and row $5$.
\end{itemize}
\end{example}

Driven by some practical considerations in space-time block codes (e.g., transceiver signal linearization, power balance, no irrational coefficients, no linear processing), Adams \etal proposed the definition of balanced complex orthogonal designs (BCOD), which turns out to be interesting on its own.

\begin{definition}
\label{def:bcod}
(Balanced complex orthogonal design \cite{Adams2011}) Complex orthogonal design $G[p, 2m, k]$ with $2m$ columns is a \emph{balanced complex orthogonal design (BCOD)} if it satisfies the followings:
\begin{itemize}
	\item Every row of $G$ has exactly $m$ zeros and $m$ nonzero entries.
	\item $G$ is conjugation-separated, i.e., all nonzero entries in any given row are either all conjugated (``a conjugated row'') or all non-conjugated (``a non-conjugated row'').
	\item For each $j \in \{1, 2, \ldots, k\}$, the submatrix $M_j$ in the $B_j$ form is skew-symmetric, i.e., $M_j^T = -M_j$.
\end{itemize}
\end{definition}

\begin{example}
\label{example:bcod_244}
$[2, 4, 4]$ BCOD.
$$
\begin{pmatrix}
z_1 & 0 & 0 & z_2\\
0 & z_1 & -z_2 & 0\\
0 & z_2^{*} & z_1^{*} & 0\\
-z_2^{*} & 0 & 0 & z_1^{*}\\
\end{pmatrix}
$$
\end{example}

\begin{example}
\label{example:bcod_844}
$[8, 4, 4]$ BCOD.
$$	
\begin{pmatrix}
z_1 & 0 & 0 & 0 & z_2 & z_3 \\
0 & z_1 & 0 & -z_2 & 0 & z_4 \\
0 & 0 & z_1 & -z_3 & -z_4 & 0 \\
0 & z_2^{*} & z_3^{*} & z_1^{*} & 0 & 0 \\
-z_2^{*} & 0 & z_4^{*} & 0 & z_1^{*} & 0 \\
-z_3^{*} & -z_4^{*} & 0 & 0 & 0 & z_1^{*} \\
z_4 & -z_3 & z_2 & 0 & 0 & 0 \\
0 & 0 & 0 & z_4^{*} & -z_3^{*} & z_2^{*} \\
\end{pmatrix}. 
$$
\end{example}

\begin{example} $[16, 8, 8]$ BCOD.
$$
\begin{pmatrix}
z_1 & 0 & 0 & 0 & 0 & z_2 & z_3 & z_4 \\
0 & z_1 & 0 & 0 & -z_2 & 0 & z_5 & z_6 \\
0 & 0 & z_1 & 0 & -z_3 & -z_5 & 0 & z_7 \\
0 & 0 & 0 & z_1 & -z_4 & -z_6 & -z_7 & 0 \\
0 & z_2^{*} & z_3^{*} & z_4^{*} & z_1^{*} & 0 & 0 & 0 \\
-z_2^{*} & 0 & z_5^{*} & z_6^{*} & 0 & z_1^{*} & 0 & 0 \\
-z_3^{*} & -z_5^{*} & 0 & z_7^{*} & 0 & 0 & z_1^{*} & 0 \\
-z_4^{*} & -z_6^{*} & -z_7^{*} & 0 & 0 & 0 & 0 & z_1^{*} \\
z_5 & -z_3 & z_2 & 0 & 0 & 0 & 0 & z_8 \\
0 & 0 & 0 & z_8^{*} & z_5^{*} & -z_3^{*} & z_2^{*} & 0 \\
z_6 & -z_4 & 0 & z_2 & 0 & 0 & -z_8 & 0 \\
0 & 0 & -z_8^{*} & 0 & z_6^{*} & -z_4^{*} & 0 & z_2^{*} \\
z_7 & 0 & -z_4 & z_3 & 0 & z_8 & 0 & 0 \\
0 & z_8^{*} & 0 & 0 & z_7^{*} & 0 & -z_4^{*} & z_3^{*} \\
0 & z_7 & -z_6 & z_5 & -z_8 & 0 & 0 & 0 \\
-z_8^{*} & 0 & 0 & 0 & 0 & z_7^{*} & -z_6^{*} & z_5^{*} \\
\end{pmatrix},
$$
\end{example}

\begin{definition} (Direct sum) Let $G_1, G_2$ be $[p_1, n, k_1], [p_2, n, k_2]$ CODs respectively on disjoint indeterminates. Define
\[
G = \begin{pmatrix}
	G_1 \\
	G_2
\end{pmatrix}
\]
to be the \emph{direct sum} of $G_1$ and $G_2$, which is a $[p_1 + p_2, n, k_1 + k_2]$ COD.
\end{definition}

\begin{definition} (Decomposable) Let $G$ be a COD. If $G$ can be expressed as the direct sum of two CODs, we say $G$ is \emph{decomposable}, and if $G$ cannot be expressed as the direct sum of two CODs, we say $G$ is \emph{indecomposable}.
\end{definition}

Using modified-Liang algorithm, Adams \etal generated $[2^m, 2m, 2^{m-1}]$ BCODs for any $m \ge 1$ \cite{Adams2011}. When $m \equiv 1,2,3 \pmod 4$, Adams \etal proved that $p \ge 2^m$ by reducing BCOD to rate-1 real orthogonal design (ROD). Liu \etal \cite{Liu2015} proved $p \ge 2^m$ when $m \equiv 0 \pmod 4$, and their proof is combinatorial. 

Using different methods, Das \cite{Das2016} also proved the lower bound $p \ge 2^m$ for a slightly more general class of CODs than BCODs. Furthermore, Das proved that BCODs of the same parameters are equivalent ignoring signs and negation.

\begin{theorem}
\label{thm:delay}
(\cite{Adams2011, Liu2015}) Let $G$ be a $[p, 2m, k]$ BCOD. Then $p \ge 2^m$.
\end{theorem}

The following two theorems are the main results of our paper. For any $m \ge 1$, we prove there is only one indecomposable BCOD with $2m$ columns up to equivalence.

\begin{theorem}
\label{thm:bcod_order}
Let $G$ be a $[p, 2m, k]$ indecomposable BCOD. Then $p = 2^m$ and $k = 2^{m-1}$.
\end{theorem}

\begin{theorem}
\label{thm:bcod_uniqueness}
For any $m \ge 1$, all $[2^m, 2m, 2^{m-1}]$ BCODs are equivalent. Moreover, we can transform one BCOD to another using equivalence operations \emph{without} column negations.
\end{theorem}

\section{Main Results}

In this section, we prove Theorem~\ref{thm:bcod_order} and Theorem~\ref{thm:bcod_uniqueness}.

%We dig into the structure of BCODs, finding that a larger BCOD can be formed by concatenating two smaller BCODs. Consequently, the parameters of BCODs can be determined using an induction process, specifically resulting in $[2^m, 2m, 2^{m-1}]$. Furthermore, for two BCODs share an identical row, we proved that they can be changed into the same BCOD through equivalence operations \emph{without} column operations by taking advantage of the symmetry and orthogonality of BCOD. Since column operations would disrupt the induction proof, this good property plays an important role in the proof of the uniqueness of BCOD. For a BCOD in standard form, we can decide on at least one row in each of the two smaller BCODs. Since equivalence operations without column operations in one smaller BCOD will not change the other BCOD, we can easily use induction to prove the uniqueness of BCOD.
%As far as we know, restriction on column operations is a new technique for proving the uniqueness. We believe that this technique has the potential for broader applications in various scenarios.

\begin{definition} (Column-restricted equivalence operations \cite{Liu2015}) A sequence of equivalence operations are called \emph{column-restricted equivalence operations} if all the column permutations are transpositions of column $i$ and column $m+i$ for some $i \in \{1, 2, \ldots, m\}$.
\end{definition}

With column-restricted equivalence operations, we can only swap column $i$ and column $m+i$ for any $i \in \{1, 2, \ldots, m\}$, instead of permuting columns arbitrarily. Besides that, other operations including row permutation, row negation, column negation, instance negation, and instance conjugation are allowed.

The following lemma says a BCOD in $B_i$ form can be put in $B_j$ form for any $j$ using column-restricted equivalence operations. This is a key lemma for understanding BCODs.

\begin{lemma}
\label{lem:liu}
(Theorem 2 in \cite{Liu2015}) Let $G$ be a $[p, n, k]$ BCOD. If $G$ is in $B_i$ form for some $i \in \{1, 2, \ldots, k\}$, then for any $j \in \{1, 2, \ldots, k\}$, $G$ can be put in $B_j$ form through column-restricted equivalence operations (without renaming the variables).
\end{lemma}

\begin{definition} In a BCOD, two rows are called \emph{a pair} if the first row is of the form
\[
(\alpha_1, \alpha_2, \ldots, \alpha_m, \beta_1, \beta_2, \ldots, \beta_m),
\]
where each $\alpha_i$ and $\beta_i$ is $0$ or $[z_r]$, and the second row is of the form
\[
(\pm \beta^*_1, \pm \beta^*_2, \ldots, \pm \beta^*_m, \pm \alpha^*_1, \pm \alpha^*_2, \ldots, \pm \alpha^*_m).
\]
\end{definition}

In Example \ref{example:bcod_244}, the first row and the third row 
\[
\begin{pmatrix}
z_1 & 0 & 0 & z_2\\
0 & z_2^{*} & z_1^{*} & 0\\
\end{pmatrix}
\]
is a pair, and the second row and the fourth row
\[
\begin{pmatrix}
0 & z_1 & -z_2 & 0\\
-z_2^{*} & 0 & 0 & z_1^{*}\\
\end{pmatrix}
\]
is a pair. In Example \ref{example:bcod_844}, after permuting rows, we can put all the rows in pairs
\[
\begin{pmatrix}
z_1 & 0 & 0 & 0 & z_2 & z_3 \\
0 & z_2^{*} & z_3^{*} & z_1^{*} & 0 & 0 \\
\hdashline
0 & z_1 & 0 & -z_2 & 0 & z_4 \\
-z_2^{*} & 0 & z_4^{*} & 0 & z_1^{*} & 0 \\
\hdashline
0 & 0 & z_1 & -z_3 & -z_4 & 0 \\
-z_3^{*} & -z_4^{*} & 0 & 0 & 0 & z_1^{*} \\
\hdashline
z_4 & -z_3 & z_2 & 0 & 0 & 0 \\
0 & 0 & 0 & z_4^{*} & -z_3^{*} & z_2^{*} \\
\end{pmatrix}.
\]

In fact, all the rows in a BCOD are in pairs.

\begin{lemma} Let $G$ be a BCOD in $B_i$ form for some $i$. Then all the rows are in pairs.	
\end{lemma}
\begin{proof}
For any $j$, in each $B_j$ form submatrix, it is clear that there are $m$ pairs of rows by Definition \ref{def:bcod}. By Lemma \ref{lem:liu}, we know if $G$ is a BCOD in $B_i$ form for some $i$, then $G$ can be transformed into $B_j$ form through column-restricted equivalence operations. Since column-restricted equivalence operations preserve pairs, we conclude that all the rows are in pairs.
\end{proof}

\begin{lemma}
\label{lem:all_diff}
Let $G$ be a $[p, n, k]$ BCOD. For any $i \in \{1,2,\ldots, k\}$, after putting $G$ into $B_i$ form, there are no repeated non-zero entries in the upper triangular part of $M_i$, where two non-zero entries are considered repeated if they are the same up to negation.
\end{lemma}
\begin{proof} Without loss of generality, let $i = 1$. We prove there are no repeated nonzero entries in the upper triangular part of $M_1$.

Assume for contradiction that there exists two repeated nonzero entries in the upper triangular part of $M_1$, denoted by $M_1(i, j)$, $M_1(s, t)$, where $i < j$, $s < t$, $i \not= s$, $j \not= t$.  Since $M_1$ is skew-symmetric, $M_1(j, i)$ and $M_1(t, s)$ are also repeated. Consider repeated nonzero entries $M_1(i, j)$ and $M_1(t, s)$, where there are two cases:
\begin{itemize}
	\item $M_1(i, j)$ and $M_1(t, s)$ are in the same row or column.
	\item $M_1(i, j)$ and $M_1(t, s)$ are in different rows and different columns.
\end{itemize}
In the former case, if nonzero entries $M_1(i, j)$ and $M_1(t, s)$ are in the same row or column (which are the same up to negation), then $G$ cannot be orthogonal, which is a contradiction! In the latter case where $M_1(i, j)$ and $M_1(t, s)$ are in different rows and different columns, consider submatrix
\[
\begin{pmatrix}
M_1(i, j) & M_1(i, s) \\
M_1(t, j) & M_1(t, s) \\
\end{pmatrix}.
\]
Since $i \not= s$ and $j \not= t$, $M_1(i, s)$ and $M_1(t, j)$ are nonzero. Since $M_1(i, j) = \pm 
	M_1(t, s)$ are also nonzero, columns $j$ and $s$ cannot be orthogonal. Contradiction!
\end{proof}

\begin{lemma}
\label{lem:removal}
Let $G$ be a $[p, 2m, k]$ BCOD in $B_i$ form for some $i \in \{1,2,\ldots,k\}$. Delete the $m^{\text{th}}$ column and $2m^{\text{th}}$ column from $G$, denoted by $G'$. $G'$ is also a BCOD.
\end{lemma}
\begin{proof} $G'$ is obviously a COD, because each $[z_i]$ appears in every column exactly once, and any two columns are orthogonal. To prove $G'$ is a BCOD, we need to verify it satisfies the 3 conditions in Definition \ref{def:bcod}. 

First, we prove every row in $G'$ contains exactly $m-1$ zeros and $m-1$ nonzeros. In the $B_i$ form submatrix, it is clear that every row contains $m-1$ zeros and $m-1$ nonzeros. For any $j$, by Lemma \ref{lem:liu}, we can transform $G'$ from $B_i$ form to $B_j$ form using column-restricted equivalence operations, meaning, the column permutations only involve swaps between column $i$ and column $m+i$. Thus, in the $B_j$ form, column $m$ and column $2m$ are removed, which implies that every row in the $B_j$ form submatrix contains half zeros and half nonzeros.

Second, we prove $G'$ is conjugation-separated, which holds obviously since $G'$ is obtained from $G$ by deleting two columns.

Finally, we prove the submatrix $M_j'$ in the $B_j$ form of $G'$ is skew-symmetric for any $j$.
Using equivalence operations without column operations, the $B_i$ form submatrix of $G'$ is equivalent to 
$
\begin{pmatrix}
z_i I_{m-1} & M_i' \\
-M'^H_i & z_i^* I_{m-1} \\
\end{pmatrix}.
$
It is clear that $M_i'$ is obtained by deleting the last column and the last row from $M_i$. Since $M_i$ is skew-symmetric, so is $M_i'$. For other $B_j$ form, we can transform $G'$ from $B_i$ form to $B_j$ form using column-restricted equivalence operations by Lemma \ref{lem:liu}. So, removing the $m^{\text{th}}$ column and $2m^{\text{th}}$ column from $G$ is the same as removing these two columns from the transformed BCOD. By the same argument as above, we can prove that $M_j'$ is skew-symmetric.
\end{proof}

Now we are ready to prove Theorem \ref{thm:bcod_order}, which says that all indecomposable BCODs are of the parameter $[2^m, 2m, 2^{m-1}]$.

%\begin{theorem}
%\label{thm:same_pk}
%Let $G$ be a $[p, 2m, k]$ indecomposable BCOD. Then $p = 2^m, k = 2^{m-1}$.
%\end{theorem}
\begin{proof} (of Theorem \ref{thm:bcod_order}) We use induction on $m$. When $m = 1$, any indecomposable BCOD with 2 columns is equivalent to
$$
\begin{pmatrix}
	z_1 && 0 \\
	0 && z_1^* \\
\end{pmatrix}.
$$

Assume the conclusion is true for $m-1$. Let us prove it for $n = 2m$. Let $G'$ be the COD by removing the $m^\text{th}$ and $2m^\text{th}$ column from $G$. By Lemma \ref{lem:removal}, we know $G'$ is a BCOD, not necessarily indecomposable.

We will prove that $G'$ is the direct sum of exactly two (indecomposable) BCODs, which would complete our proof, since all indecomposable BCODs in $2(m-1)$ columns have parameters $[2^{m-1}, 2(m-1), 2^{m-2}]$ by induction hypothesis.

Without loss of generality, we assume $G$ is in $B_1$ form. Since $G'$ is obtained by deleting the $m^\text{th}$ and $2m^{\text{th}}$ column from $G$, we can write
\begin{equation}
G'=
\begin{pmatrix}
z_1I_{m-1} & M'_1 \\
-{M'}_1^H & z^*_1 I_{m-1} \\
\ldots & \ldots \\
\end{pmatrix}
\end{equation}
and
\begin{equation}
\label{equ:G_B1form}
G=
\begin{pmatrix}
z_1I_{m-1} & 0 & M'_1 & \alpha \\
0 & z_1 & -\alpha^T & 0 \\
-{M'}_1^H & \alpha^* & z^*_1 I_{m-1} & 0 \\
-\alpha^H & 0 & 0 & z_1^*\\
\ldots & \ldots & \ldots & \ldots \\
\end{pmatrix},
\end{equation}
where $\alpha$ is a vector with $m-1$ nonzero entries. By Lemma \ref{lem:all_diff},  all the entries in $\alpha$ are distinct and different from those in $M_1'$. By Lemma \ref{lem:removal}, BCOD $G'$ is a direct sum of indecomposable BCODs, denoted by $H_1, H_2, \ldots, H_{\ell}$, where $\ell \ge 1$.

Without loss of generality, assume that $H_1$ contains 
$
\begin{pmatrix}
z_1I_{m-1} & M'_1 \\
-{M'}_1^H & z^*_1 I_{m-1} \\
\end{pmatrix},
$
and $H_1$ or $H_2$ contains row $(\underbrace{0, \ldots, 0}_{m \text{ times}}, \alpha^T)$, whose existence can be seen from the second line in the matrix in \eqref{equ:G_B1form}. Putting them together, we claim the first part of $G$ must be of the following form:
$$
 \begin{pmatrix}
z_1I_{m-1} & 0 & M'_1 & \alpha \\
0 & z_1 & -\alpha^T & 0 \\
-{M'}_1^H & \alpha^* & z^*_1 I_{m-1} & 0 \\
-\alpha^H & 0 & 0 & z_1^*\\
\end{pmatrix}.
$$
Append the rest of rows in $H_1$ (as well as $H_2$, if row $(\underbrace{0, \ldots, 0}_{m \text{ times}}, \alpha^T)$ belongs to $H_2$) to the above matrix, we have
\begin{equation}
\label{equ:BCOD_completion}
 \begin{pmatrix}
z_1I_{m-1} & 0 & M'_1 & \alpha \\
0 & z_1 & -\alpha^T & 0 \\
-{M'}_1^H & \alpha^* & z^*_1 I_{m-1} & 0 \\
-\alpha^H & 0 & 0 & z_1^*\\
L & ? & R & ?\\
\end{pmatrix},
\end{equation}
where $(L, R)$ denotes the rest of rows in $H_1$ (and $H_2$), and $?$ denotes undetermined entries.

Our goal is to prove
\begin{itemize}
	\item  all the undetermined entries (marked by ?) can be fully determined by $L, R$ and the $B_1$ form submatrix,
	\item  there are no new nonzero entries except those in $L$ and $R$, and
	\item  column $m$ and column $2m$ \emph{exhaust} all indeterminates, including $[z_1]$, and all indeterminates in $\alpha$, $L$ and $R$. 
\end{itemize}
Note that the above matrix is a submatrix of a BCOD $G$. Together the above implies that \eqref{equ:BCOD_completion} is indeed a BCOD, which implies $\ell = 1$ or $2$. By Theorem \ref{thm:delay}, $\ell = 1$ is impossible, otherwise there exists a $[2^{m-1}, 2m, 2^{m-2}]$ BCOD. So $\ell = 2$, which will complete our induction.

\vspace{0.2cm}
Now let us see how to determine the entries in column $m$ and $2m$ (marked by $?$) step by step. If a row contains undetermined entries (marked as $?$), say the row is \emph{undetermined}; otherwise, say the row is \emph{determined}. During the process, we always determine the rows in \emph{pairs}. As such, column $m$ and column $2m$ always contain the same set of indeterminates.

If column $m$ and column $2m$ do not exhaust all indeterminates (in $L$ and $R$), there must exist a $[z_r]$ such that $[z_r]$ has not appeared in column $m$ and $2m$, and $[z_r]$ has appeared in a pair of determined rows. (Otherwise, it would contradict the indecomposability of $G$.) Furthermore, the determined rows, say row $i_1$ and $i_2$, that contain $[z_r]$ is of the form
\begin{eqnarray*}
\text{row $i_1$: } & (\underbrace{\ldots, [z_r], \ldots, [z_s]}_{m\text{ entries}}, \underbrace{\ldots, 0}_{m\text{ entries}}), \\
\text{row $i_2$: } & (\underbrace{\ldots, 0}_{m\text{ entries}}, \underbrace{\ldots, [z_r], \ldots, [z_s]}_{m\text{ entries}}), \\
\end{eqnarray*}
or
\begin{eqnarray*}
\text{row $i_1$: } & (\underbrace{\ldots, [z_r], \ldots, 0}_{m\text{ entries}}, \underbrace{\ldots, [z_s]}_{m\text{ entries}}), \\
\text{row $i_2$: } & (\underbrace{\ldots, [z_s]}_{m\text{ entries}}, \underbrace{\ldots, [z_r], \ldots, 0}_{m\text{ entries}}). \\
\end{eqnarray*}
Without loss of generality, assume the former case, where $[z_r]$ is on row $i_1$ column $j$, for the other case is similar. By the definition of COD, indeterminate $[z_s]$ must appear on column $j$, say $G(i_1', j) = [z_s]$. Consider the $2 \times 2$ submatrix $G(i_1, i_1'; j, m)$, which is
\[
\begin{pmatrix}
[z_r] & [z_s] \\
[z_s] & ? \\
\end{pmatrix},
\]
then $?$ must be $[z_r]$, where the sign and conjugation can be determined due to orthogonality. The other $?$ on column $2m$ must be zero by the definition of BCOD (i.e., Definition \ref{def:bcod}). Similarly, indeterminate $[z_s]$ must appear on column $m+j$, say $G(i_2', m+j) = [z_s]$. Consider the $2 \times 2$ submatrix $G(i_2, i_2'; j+m, 2m)$, which is
\[
\begin{pmatrix}
[z_r] & [z_s] \\
[z_s] & ? \\
\end{pmatrix}.
\]
By orthogonality, $?$ must be $[z_r]$, whose sign and conjugation can be uniquely determined. The other $?$ on column $m$ must be zero by the definition of BCOD. In this way, we have determined a pair of rows that contains $[z_r]$ (on column $m$ and column $2m$).

Repeat the above process until column $m$ and column $2m$ exhaust all indeterminates. We claim all the rows are determined. Otherwise, if there exists an undetermined row, say
\[
\text{row $i_1$: } (\underbrace{[z_r], \ldots, ?}_{m\text{ entries}}, \underbrace{0, \ldots, ?}_{m\text{ entries}})
\]
or 
\[
\text{row $i_1$: } (\underbrace{0, \ldots, ?}_{m\text{ entries}}, \underbrace{[z_r], \ldots, ?}_{m\text{ entries}}).
\]
Assume the former case without loss of generality. By our assumption that column $m$ and column $2m$ have exhausted all indeterminates, $[z_r]$ must have appeared on column $m$ and column $2m$, say $G(i_1', m) = [z_r]$. Consider the $2 \times 2$ submatrix $G(i_1, i_1'; 1, m)$, which is
\begin{equation}
\label{equ:bcod_order_last1}
\begin{pmatrix}
[z_r] & ? \\
0 & [z_r] \\
\end{pmatrix}
\end{equation}
or 
\begin{equation}
\label{equ:bcod_order_last2}
\begin{pmatrix}
[z_r] & ? \\
[z_s] & [z_r] \\
\end{pmatrix}
\end{equation}
If \eqref{equ:bcod_order_last2} happens, then $G(i_1, m)$ must be $[z_s]$ due to orthogonality, which contradicts the assumption that column $m$ and column $2m$ exhaust all indeterminates. So we must have \eqref{equ:bcod_order_last1}, which implies that $G(i_1, m) = 0$ by orthogonality. In this way, we can prove both entries $G(i_1, m)$ and $G(i_1, 2m)$ (marked by ?) must be 0, which contradicts to the definition of BCOD.
\end{proof}

From the proof the above theorem, we have

\begin{corollary}
\label{cor:columnm2m_unique}
Let $G$ be a $[2^m, 2m, 2^{m-1}]$ BCOD. If the first $2m$ rows of $G$ are the $B_i$ form submatrix, and columns $1, 2, \ldots, m-1$ and $m+1, m+2, \ldots, 2m$ are known, then columns $m$ and column $2m$ can be uniquely determined.
\end{corollary}

We have shown that when the number of columns is fixed, all indecomposable BCODs have the same parameter. In fact, they are all equivalent (under equivalence operations). We will prove this result using induction on a stronger statement.

\begin{theorem}
	Let $G_1, G_2$ be two $[2^m, 2m, 2^{m-1}]$ BCODs. If $G_1$ and $G_2$ have the same pair of rows, then $G_1$ can be transformed to $G_2$ by equivalence operations without any column operations (including column permutations and column negations).
\end{theorem}
\begin{proof} We use induction on $m$. When $m = 1$, the conclusion is trivial.

Suppose the conclusion is true for $m-1$. Without loss of generality, assume the first two rows of $G_1$ and $G_2$ is
$$
\begin{pmatrix}
z_1 & 0 & \cdots & 0 & 0 & 0 & z_2 & \cdots & z_{k-1} & z_k\\
0 & z_2^* & \cdots & z_{k-1}^* & z_k^* & z_1^* & 0 & \cdots & 0 & 0 
\end{pmatrix},
$$
and thus both $G_1$ and $G_2$ are in $B_1$ form. Why? Because indeterminate $\pm z_1$ must appear on column $1, 2, \ldots, m$; indeterminate $\pm z^*_1$ must appear on column $m, m+1, \ldots, 2m$.

Our goal is to prove $G_1$ and $G_2$ can be transformed to the same COD without column operations, which implies that $G_1$ can be transformed to $G_2$ by equivalent operations without column operations.

Note that after removing column $m$ and $2m$, $G_1$ is the direct sum of two $[2^{m-1}, 2(m-1), 2^{m-2}]$ BCODs, denoted by $M_1$ and $N_1$; $G_2$ is the direct sum of $M_2$ and $N_2$. By induction hypothesis, $M_1$ can be transformed to $M_2$ without using any column operations, since they share the same pair of rows
$$
\begin{pmatrix}
z_1 & 0 & \cdots & 0 & 0 & z_2 & \cdots & z_{k-1}\\
0 & z_2^* & \cdots & z_{k-1}^* & z_1^* & 0 & \cdots & 0 
\end{pmatrix}.
$$

Since $G_1$ is in $B_1$ form, using equivalence operations without column operations, $G_1$ is equivalent to
$$
\begin{pmatrix}
z_1I_{m-1} & 0 & M & \alpha \\
0 & z_1 & -\alpha^T & 0 \\
-M^H & \alpha^* & z_1^*I_{m-1} & 0\\
-\alpha^H & 0 & 0 & z_1^*\\
M'_{1L} & * & M'_{1R} & *\\
N'_{1L} & * & N'_{1R} & *\\
\end{pmatrix},
$$
where $(M'_{1L}, M'_{1R})$ denotes the rest of rows in $M_1$ excluding 
$
\begin{pmatrix}
z_1I_{m-1} & M \\
-M^H & z_1^*I_{m-1}\\
\end{pmatrix},
$
and $(N'_{1L}, N'_{1R})$ denotes the rest of rows in $N_1$ excluding
$
\begin{pmatrix}
0 & -\alpha^T \\
-\alpha^H & 0\\
\end{pmatrix}.
$
Similarly, using equivalence operations without column operations, we can transform $G_2$ into the form
$$
\begin{pmatrix}
z_1I_{m-1} & 0 & M & \alpha \\
0 & z_1 & -\alpha^T & 0 \\
-M^H & \alpha^* & z_1^*I_{m-1} & 0\\
-\alpha^H & 0 & 0 & z_1^*\\
M'_{2L} & * & M'_{2R} & *\\
N'_{2L} & * & N'_{2R} & *\\
\end{pmatrix},
$$
where the first $2m$ rows are exactly the same as $G_1$. For example, we use instance renaming to make the indeterminates in $\alpha$ the same. If the conjugation or negation does not match, we can conjugate or negate the indeterminate, since it appears for the first time. If $z_1$ on the second row is negated, we can negate the entire row. In this way, we have a pair of rows
$$
\begin{pmatrix}
0 & -\alpha^T \\
-\alpha^H  & 0 \\
\end{pmatrix}
$$
that has also appeared in $N_1$. By the induction hypothesis, $N_2$ can be transformed to $N_1$ without using column operations.

By Corollary \ref{cor:columnm2m_unique}, the remaining entries on column $m$ and $2m$ are uniquely determined, including signs and conjugations. 
\end{proof}

\section{Conclusion}

We have proved that the parameters of an indecomposable BCOD are uniquely determined by the number of columns. More precisely, for any indecomposable BCOD with $2m$ columns, the parameters must be $[2^m, 2m, 2^{m-1}]$. Furthermore, we have proved that any two indecomposable BCODs of the same parameters are equivalent (as of Definition~\ref{def:equiv}). In other words, there is only one indecomposable BCOD with $2m$ columns up to equivalence. 

So far, we have fully understood the structure of BCODs; as shown by Adams \etal \cite{Adams2011}, for any $m$, one can generate a $[2^m,2m,2^{m-1}]$ BCOD by using the modified-Liang algorithm. For general COD classes, there is still a lot of work to do. The main problem, that is, deciding all possible parameters of CODs, is wide open. We hope our methods will shed light on more general cases.

%So far, we have fully understood the structure of BCODs, and we can generate a $[2^m,2m,2^{m-1}]$ BCOD by using the modified-Liang algorithm in \cite{Adams2011}. 
%For more general cases, there are still lots of work to do. Especially the main problem of complex orthogonal design that deciding the parameters $p,n,k$ to ensure a $[p,n,k]$ COD exists is open.
%The new methodology we used in this paper is that we restrict column operations by fixing a pair of rows, which simplifies the problem. We believe this method can be used in more general cases to help solve the main problem of complex orthogonal design. For instance, if we can establish the uniqueness of certain rows in a matrix, we can reduce the main problem from the entire matrix to those specific rows.

\printbibliography

\end{document}